\documentclass[11pt]{article}
\usepackage{amsfonts,amsmath,amssymb,amsthm}
\usepackage{mathtools}
\usepackage{derivative}
\usepackage{fullpage}
\usepackage[colorlinks]{hyperref}
\usepackage[nameinlink,capitalise]{cleveref}
\usepackage{xcolor}

\hypersetup{
  linkcolor=[rgb]{0,0,0.6},
  citecolor=[rgb]{0, 0.4, 0},
  urlcolor=[rgb]{0.6, 0, 0}
}

\newif\ifmynotes
\mynotestrue

\newcommand*\samethanks[1][\value{footnote}]{\footnotemark[#1]}

\theoremstyle{definition}
\newtheorem{thm}{Theorem}[section]
\newtheorem{defn}[thm]{Definition}
\newtheorem{lem}[thm]{Lemma}
\newtheorem{cor}[thm]{Corollary}

\newtheorem{question}[thm]{Question}
\newtheorem{fact}[thm]{Fact}
\newtheorem{construction}[thm]{Construction}
\theoremstyle{remark}

\newcommand{\NN}{\mathbb{N}}

\newcommand{\FF}{\mathbb{F}}
\newcommand{\KK}{\mathbb{K}}
\newcommand{\Res}[1]{\mathrm{Res}\left(#1\right)}
\newcommand{\coeff}[1]{\mathrm{coeff}\left(#1\right)}
\renewcommand{\epsilon}{\varepsilon}
\newcommand{\eps}{\varepsilon}

\title{Optimal Pseudorandom Generators for Low-Degree Polynomials Over Moderately Large Fields}
\author{Ashish Dwivedi\thanks{
Department of Computer Science and Engineering, The Ohio State University. Emails: \texttt{ashish02dwivedi@gmail.com, zguotcs@gmail.com}.
}
\and
Zeyu Guo\samethanks[1]
\and
Ben Lee Volk\thanks{Efi Arazi School of Computer Science, Reichman University, Israel. Email: \texttt{benleevolk@gmail.com}. The research leading to these results has received funding from the Israel Science Foundation (grant number 843/23).}
}
\date{}

\begin{document}

\maketitle

\begin{abstract}
We construct explicit pseudorandom generators that fool $n$-variate polynomials of degree at most $d$ over a finite field $\mathbb{F}_q$. The seed length of our generators is $O(d \log n + \log q)$, over fields of size exponential in $d$ and characteristic at least $d(d-1)+1$. Previous constructions such as Bogdanov's (STOC 2005) and Derksen and Viola's (FOCS 2022) had either suboptimal seed length or required the field size to depend on $n$.

Our approach follows Bogdanov's paradigm while incorporating techniques from Lecerf's factorization algorithm (J. Symb. Comput. 2007) and insights from the construction of Derksen and Viola regarding the role of indecomposability of polynomials.
\end{abstract}

\section{Introduction}
\label{sec:intro}
The role of randomness in efficient computation is one of the central topics in complexity theory: random bits are useful for designing algorithms, but producing random bits comes at a cost and it is often desirable to reduce them or eliminate them altogether.
One of the simplest yet most profound insights in this area is that efficient algorithms are, by definition, computationally limited, and cannot perform arbitrary statistical tests over their random bits. Therefore, one may hope to construct \emph{pseudorandom} distributions that use less random bits but are able to ``fool'' some limited classes of tests, that cannot distinguish between them and between truly random bits.

For the pseudorandom distributions to be useful, they need to be efficiently computable themselves. This is usually modeled as a \emph{pseudorandom generator} (PRG, for short). A PRG for a class of $\mathcal{C}$ is an efficiently computable function $G : S \to B$ such that for every function $f \in \mathcal{C}$, the distributions $f(\mathbf{U}_B)$ and $f(G(\mathbf{U}_S))$ are close in statistical distance, where $\mathbf{U}_A$ denotes the uniform distribution over the set $A$. Namely, the two experiments of applying $f(\cdot)$ to a uniformly random element of $B$, and applying $f(G(\cdot))$ to a uniformly random element of $S$, give roughly the same results. For this to be useful and non-trivial, obviously the set $S$ needs to be significantly smaller than $B$. The quantity $\log |S|$ is called the \emph{seed length} of the generator.

There has been a significant amount of work on constructing pseudorandom generators for various types of restricted distinguishers. In its most general form, where the distinguisher is allowed to be an arbitrary efficient (even non-uniform) algorithm, constructing such PRGs would imply breakthrough lower bounds in complexity theory. However, there are also unconditional constructions of PRGs for distinguishers coming from certain smaller complexity classes (see, for example, the surveys \cite{Vadhan12, HH23}).

In this paper, we focus on pseudorandom generators in the algebraic setting. Here, the restriction on the distinguishers is of algebraic nature: we seek to fool distinguishers that are low-degree $n$-variate polynomials over finite fields.

The problem of fooling low-degree polynomials is well-studied. The most basic case is polynomials of degree one, i.e., fooling linear functions. Such generators are also known as \emph{$\varepsilon$-biased sets}, and this problem was traditionally studied over $\mathbb{F}_2$, although some of the constructions can be generalized to larger fields. This concept was first defined and considered by Naor and Naor \cite{NN93}, with various improved constructions given by \cite{AGHP92, EGLNV98, BT13}, culminating in a recent nearly-optimal construction by Ta-Shma \cite{Ta-Shma17}. The seed length in those constructions is $O(\log n + \log q + \log(1/\eps))$, where $n$ denotes the number of variables, $\varepsilon$ the error of the PRG, and $q$ the field size.

While focusing on polynomials of degree one might seem a bit too restrictive, $\varepsilon$-biased sets have found numerous applications throughout the field of pseudorandomness and derandomization, and in the theory of computation in general.

One example relevant to this work is that $\varepsilon$-biased sets are in fact a basic building block in a construction of PRGs for higher-degree polynomials, using a paradigm initiated by Bogdanov and Viola \cite{BV10}. They suggested constructing a generator for degree-$d$ polynomials by summing up $\ell=\ell(d)$ independent copies of a generator for degree-one polynomials. The paper \cite{BV10} proved a conditional result when the number of summands is $d$, assuming certain additive combinatorics conjectures. Lovett \cite{Lovett09} showed how to prove an unconditional result at the cost of making the number of summands $2^d$. Finally, Viola \cite{Viola09} showed (unconditionally) that in fact $d$ summands suffice. The seed length in his construction is $O(d \log n + d 2^d \log(q/\eps))$. Indeed, even though the construction only sums $d$ copies of a generator for degree-one polynomials, for the analysis to go through, the error of this generator needs to be as small as $\varepsilon^{2^d}$ (for the final error of the generator for degree-$d$ polynomials to be $\varepsilon$), which incurs a factor of $2^d$ in the final seed length. Improving this generator and in particular obtaining meaningful results for polynomials of degree greater than $\log n$ is an extremely important open problem in complexity theory. One reason is that such pseudorandom generators will yield pseudorandom generators for small constant-depth circuits with parity gates, since Razborov \cite{Razborov87} and Smolensky \cite{Smolensky93} famously proved that functions computed by such circuits are approximated by low-degree polynomials.

All the constructions mentioned above work for any field. There are, however, better results when the field size $q$ is assumed to be large (typically, at least polynomially large in $d$ and $1/\varepsilon$). This assumption is useful since it allows one to use powerful tools from algebraic geometry, such as Weil-type estimates \cite{Wei49} on the number of points of varieties over finite fields.

This line of work was initiated by Bogdanov \cite{Bogdanov05}, who showed how to use different pseudorandom objects called \emph{hitting set generators} for low-degree polynomials in order to construct pseudorandom generators. Bogdanov's work, followed by the later improved constructions of hitting set generators \cite{KS01, Lu12, CT13, GX14}, resulted in a PRG with seed length $O(d^4 \log n + \log q)$ assuming $q \ge C d^6/\varepsilon^2$ for a sufficiently large constant $C$.\footnote{One should note that since $q$ is polynomially large in $1/\varepsilon$, the seed length also implicitly depends on $\log (1/\varepsilon)$ through the $\log q$ term.} We expand more on Bogdanov's technique in Section \cref{sec:technique}, as it is very relevant to this work. 

More recently, Derksen and Viola \cite{DV22} introduced fundamentally new techniques for this problem, with tools coming from invariant theory. One of their key ideas is to construct a low-degree polynomial map on a few variables that preserves the \emph{indecomposability} property of a polynomial $f$ when composed with it (we refer to \cref{sec:prelim} for more on that). 
Using this new tool in conjunction with other techniques, they are able to construct generators with seed length $O(d \log (dn) + \log q)$, assuming $q \ge C d^4 n^{\delta}/\varepsilon^2$ (for some large constant $C$ and small constant $\delta$), or seed length $O(d \log n \log (d\log n) + \log q)$ for $q \ge C (d \log n)^4/\varepsilon^2$. One should note, however, that the optimal parameters in the construction of \cite{DV22} are obtained after composing their construction with Bogdanov's original construction.

A related natural question is how small the seed length can potentially be. Alon, Ben-Eliezer and Krivelevich \cite{ABK08} considered this question and proved a lower bound of $\Omega(d \log(n/d) + \log q + \log(1/\eps))$ on the seed length. Thus, we see that the explicit constructions of \cite{DV22} come very close to the optimal bound. 
However, unlike the result of Bogdanov \cite{Bogdanov05}, in the construction of Derksen and Viola \cite{DV22} the minimum field size depends on the number of variables $n$. 

\subsection{Our Results}
\label{sec:results}

In this paper, we provide an improved construction of PRGs for low-degree polynomials, with an even shorter seed length, assuming the field size is exponentially large in $d$ (but independent of $n$).

\begin{thm}
\label{thm:main}
Let $\mathbb{F}_q$ be a finite field of characteristic at least $d(d-1)+1$ and size $q \ge C (d2^d/\eps+d^4/\eps^2)$ (for some sufficiently large absolute constant $C$). Then,
there exists an explicit pseudorandom generator that fools $n$-variate polynomials of degree at most $d$ over $\mathbb{F}_q$ with error $\epsilon$ and seed length $O(d\log n + \log q)$.
\end{thm}

For convenience, we summarize the comparison between  \cref{thm:main} and the results of Bogdanov \cite{Bogdanov05}, Viola \cite{Viola09} and Derksen and Viola \cite{DV22} in the following table. All the entries in this table are given up to some constant factors, but for ease of readability, we omit $O(\cdot)$ notations. 

\begin{center}
\begin{tabular}{l|c|c}
 & \textbf{Seed Length} & \textbf{Field Size} \\ \hline
 \cite{Viola09} & $d \log n + d \cdot 2^d \log (q/\varepsilon)$ & Every $q \ge 2$ \\
 \cite{Bogdanov05}+\cite{GX14} & $d^4 \log n + \log q$ & $d^6/\varepsilon^2$ \\ 
 \cite{DV22} & $d\log(dn) + \log q$ & $d^4n^{0.001}/\varepsilon^2$ \\  
\cite{DV22} & $d \log n \cdot \log(d \log n) + \log q$ & $(d \log n)^4 / \varepsilon^2$ \\  
 This paper: & $d \log n + \log q$ & $d2^d/\eps+d^4/\eps^2$    
\end{tabular}
\end{center}

We also prove that, if we only want to fool polynomials of \emph{prime} degree up to $d$, then the required field size in \cref{thm:main} can be improved to $O(d^4/\eps^2)$, avoiding an exponential dependence on $d$.

\begin{thm}\label{thm:prime-degree-intro}
Let $\mathbb{F}_q$ be a finite field of characteristic at least $d(d-1)+1$ and size $q \ge C (d^4/\eps^2)$ (for some sufficiently large absolute constant $C$). Then, there exists an explicit pseudorandom generator that fools $n$-variate polynomials of prime degree up to $d$ over $\mathbb{F}_q$ with error $\epsilon$ and seed length $O(d\log n + \log q)$.
\end{thm}

\subsection{Proof Techniques}
\label{sec:technique}

We start by reviewing the proof of Bogdanov's PRG.
Bogdanov's idea is to consider restrictions of the polynomial $f$ we are trying to fool onto planes, and to argue that ``most'' planes preserve the output distribution of the polynomial. Since a plane is a two-dimensional subspace, after having selected a good plane, we only need to sample two more field elements to select a random element from the plane.

The question now is how to find a good plane. Here, Bogdanov uses results by Kaltofen \cite{Kaltofen95}, who proved an effective version of Hilbert's irreducibility theorem. Kaltofen demonstrated that, for every degree-$d$ irreducible polynomial $f$, there exists a polynomial $P$ of degree roughly $d^4$, whose variables correspond to the parameters of the plane, such that every point at which $P$ is nonzero corresponds to a ``good'' plane for $f$---that is, a plane that preserves the irreducibility of $f$. Therefore, we can use a hitting set generator for polynomials of degree $d^4$ to find a good plane. This results in a factor of $d^4$ in the final seed length.

Over fields of large characteristic (or characteristic zero), Lecerf \cite{Lec06, Lec07} obtained an improved upper bound of $O(d^2)$ on the degree of such polynomials, based on ideas of Ruppert \cite{Ruppert86, Ruppert99} and Gao \cite{Gao2003}. However, Lecerf also presents an example where the degree of such a polynomial must be at least $\Omega(d^2)$, demonstrating that this approach alone may not suffice to achieve an improvement within Bogdanov's framework.

Derksen and Viola circumvent this problem by using a different approach based on preserving the \emph{indecomposability} of polynomials: a polynomial $f(x_1, \ldots, x_n)$ is indecomposable if it cannot be written as $f=g(h(x_1, \ldots, x_n))$ where $g$ is a univariate polynomial of degree at least $2$. We refer to \cref{sec:prelim} for precise definitions.

To prove our result, we revisit Bogdanov's approach, while noting that by applying Lecerf's results \cite{Lec07} (rather than Kaltofen's bounds \cite{Kaltofen95}) in a more careful way, and assuming the field is sufficiently large, it is in fact enough to use hitting sets generators only for polynomials of degree $O(d)$, rather than $O(d^2)$. This requires following the outline above but making sure that at each step, we only need to hit polynomials of degree $O(d)$.

\sloppy
Following Lecerf's notation and terminology, suppose $F(x_1, \ldots, x_n,y)$ is an irreducible polynomial. A point $\mathbf{a}=(a_1, \ldots, a_n)$ is called a \emph{Bertinian good} point if the bivariate polynomial $H(x,y)=F(a_1 x, \ldots, a_n x, y)$ remains irreducible. Lecerf proves that there exists a polynomial $\mathcal{A}(z_1, \ldots, z_n)$ of degree $O(d^2)$ such that if $\mathcal{A}(a_1, \ldots, a_n) \neq 0$ then $\mathbf{a}$ is a Bertinian good point. This is achieved by transforming the question of irreducibility into a question about the rank of a solution space for a certain linear system that depends on $a_1, \ldots, a_n$ (see  \cref{sec:lecerf}). This transformation naturally leads to defining $\mathcal{A}$ as a certain minor of the matrix representing that linear system. The minor has dimensions $O(d) \times O(d)$, and each entry of the matrix is a polynomial of degree $O(d)$, which results in a total bound of $O(d^2)$ on the degree of its determinant, $\mathcal{A}$.

We, however, observe that Lecerf's results actually imply a much stronger structure of this linear system: its solution space, over any field, is always spanned by vectors whose entries are in $\{0,1\}$. In fact, Lecerf directly characterizes the relationship between the irreducible factors of $F$ and the vectors spanning the solution space, though this detail is irrelevant for the moment.

Thus, in Lecerf's argument, $a_1, \ldots, a_n$ are chosen such that a particular minor is nonzero, namely, a certain linear system has no non-trivial solutions. But it is enough to select $a_1, \ldots, a_n$ in a way that only guarantees that this linear system has no non-trivial $0/1$ solutions!

Fixing any vector $u \in \{0,1\}^d$, the requirement that $u$ is not a solution to the linear system turns out to be a condition expressible as $u$ being a nonzero of a polynomial of degree $O(d)$, rather than $O(d^2)$. If we pick $a_1, \ldots, a_n$ from a hitting set generator with error $\delta$ smaller than $2^{-d}$, we can afford to take a union bound over all vectors in $\{0,1\}^d$ and ensure that none of them is a solution to the linear system while keeping the total error small. This is not a big price to pay in terms of the seed length of the HSG, which is $O(d \log n + \log (1/\delta))$, so requiring $\delta$ to be exponentially small in $d$ adds an insignificant additive $O(d)$ term. Where we do pay the price for the small error is in the field size, since the explicit construction of the HSG we use requires the field size to be at least roughly $d/\delta$. Fortunately, however, the dependence of the seed length of our generator on the field size $q$ is also by an additive $O(\log q)$ term, which means that once more requiring $q$ to be exponentially large in $d$ has no adverse effects even on the total seed length of the PRG.

We briefly remark that, for technical reasons, Lecerf's result also requires the characteristic of the underlying field to be zero or at least $d(d-1)+1$. We further elaborate on Lecerf's techniques in  \cref{sec:lecerf}.

We finally mention another technical point. Lecerf's irreducibility characterization \cite{Lec07} assumes a technical condition on the polynomial, which he called Hypothesis (H) (see \cref{sec:H}). Such a ``preprocessing'' step, which makes the polynomial monic in a certain distinguished variable, is common to many factorization algorithms, and can usually be easily guaranteed by applying a random linear transformation to the variables. However, doing this in the na\"{i}ve way would require the use of too many random bits. To solve this problem, in \cref{sec:H} we show that this part can also be derandomized by using a hitting set generator for polynomials of degree $O(d)$.

\section{Preliminaries}
\label{sec:prelim}

We now define the basic objects studied in this paper and introduce the fundamental mathematical concepts used.

\paragraph{Notations.}

All logarithms are base $2$. 
Denote by $\NN$ the set of natural numbers $\{0,1,2,\dots\}$.
For $n\in\NN$, define $[n]=\{1,2,\dots,n\}$. 
For a finite set $A$, denote by $\mathbf{U}_A$ the uniform distribution over $A$.

We often use symbols in bold, e.g., $\mathbf{a}$ or $\mathbf{x}$, as the shorthand for a vector $(a_1,\dots,a_n)$ or a sequence of variables $x_1,\dots,x_n$.

Denote by $\FF_q$ the finite field of size $q$. The algebraic closure of a field $\FF$ is denoted by $\overline{\FF}$.
For a commutative ring $A$ and variables $x_1,\dots,x_n$, we denote by $A[[x_1,\dots,x_n]]$ or $A[[\mathbf{x}]]$ the \emph{ring of formal power series} over $A$ in $x_1,\dots,x_n$, i.e., 
\[
A[[\mathbf{x}]]=\left\{\sum_{\mathbf{e}=(e_1,\dots,e_n)\in\NN^n} a_{\mathbf{e}} x_1^{e_1}\cdots x_n^{e_n}: a_{\mathbf{e}}\in A\right\}.
\]

\paragraph{Pseudorandom Generators and Hitting Set Generators.}

\begin{defn}[Pseudorandom generator, PRG]
\label{def:PRG}
Let $\mathbb{F}_q$ be a finite field. A \emph{pseudorandom generator (PRG)} for $n$-variate polynomials of degree at most $d$ over $\mathbb{F}_q$ with error $\varepsilon$ is an efficiently computable map $G : S \to \mathbb{F}_q^n$ from a finite set $S\neq \emptyset$ such that for every such polynomial $f$ of degree at most $d$, the two distributions $f(G(\mathbf{U}_{S}))$ and $f(\mathbf{U}_{\mathbb{F}_q^n})$ are $\varepsilon$-close in statistical distance. That is,
\[
\frac{1}{2} 
\sum_{a \in \mathbb{F}_q}
\left|
\Pr_{\mathbf{x} \in \mathbb{F}_q^n} [f(\mathbf{x})=a] - \Pr_{\mathbf{y} \in S} [f(G(\mathbf{y}))=a]
\right| \leq \varepsilon.
\]
The quantity $\log |S|$ is called the \emph{seed length} of $G$.
\end{defn}

A weaker object than a PRG is a \emph{hitting set generator}. Here, we only require that a nonzero polynomial is nonzero (with high probability) on the output of the generator. 

\begin{defn}[Hitting set generator, HSG]
\label{def:HSG}
Let $\FF$ be a field. A \emph{hitting set generator (HSG)} with density $1-\delta$ for $n$-variate polynomials of degree at most $d$ over $\FF$ is an efficiently computable map $H : S \to \FF^n$ from a finite set $S\neq\emptyset$ such that for every such nonzero polynomial $f$ of degree at most $d$,
\[
\Pr_{\mathbf{y} \in S} [f(H(\mathbf{y}))=0] \leq \delta.
\]
The quantity $\log |S|$ is called the \emph{seed length} of $G$.
\end{defn}

Building on the earlier work \cite{KS01, Lu12} and algebraic-geometric codes, Guruswami and Xing \cite{GX14} constructed explicit HSGs for low-degree polynomials with asymptotically optimal seed length and density.

\begin{thm}[\cite{GX14}]\label{thm:optimalHSG}
There exists an absolute constant $C$ such that 
for any $n,d,q,\delta$, such that $q \geq Cd/\delta$,
there exists an explicit HSG for $n$-variate polynomials of degree at most $d$ over $\FF_q$ with density $1-\delta$ and seed length $O(d\log n + \log (1/\delta))$.
\end{thm}

It should also be noted that for hitting set generators, the field $\FF$ does not have to be finite. This generality is used in the statement of the following fact, that an HSG for a field $\FF$ is also a HSG for any extension field $\KK$ of $\FF$.

\begin{fact}[\cite{Bogdanov05, DV22}]\label{fact:extension}
Let $H: S\to\FF$ be an HSG with density $1-\delta$ for polynomials of degree at most $d$ over a field $\FF$, and let $\KK$ be an extension of $\FF$. Then $H$ is also an HSG with density $1-\delta$ for polynomials of degree at most $d$ over $\KK$.
\end{fact}

\begin{proof}
Let $\mathcal{B}$ be a basis of $\KK$ over $\FF$, and let $f$ be a nonzero polynomial in $\KK[x_1, \ldots, x_n]$ of degree at most $d$. By expressing every coefficient $c \in \KK$ of a monomial in $f$ as a linear combination $c = \sum_{b \in \mathcal{B}} a_b \cdot b $ with $a_b \in \FF$ for every $b \in \mathcal{B}$, we may write $f = \sum_{b \in \mathcal{B}} f_b \cdot b$ such that $f_b \in \FF[x_1, \ldots, x_n]$ is a polynomial of degree at most $d$ for every $b \in \mathcal{B}$, and at least one $f_b$ is nonzero. Thus, for any $u \in S$, $f(H(u)) = \sum_{b \in \mathcal{B}} f_b(u) \cdot b$ is nonzero unless $f_b(H(u))=0$ for every $b$, which happens with probability at most $\delta$ over the choice of $u \in S$.
\end{proof}

\paragraph{Indecomposable Polynomials.}

The \emph{indecomposability} of a polynomial is crucially used in the analysis of the PRG construction in \cite{DV22} as well as in our analysis. We first define this property.

\begin{defn}[Indecomposability]
Let $f\in\FF[\mathbf{x}]$ be a non-constant polynomial over a field $\FF$. It is said to be \emph{decomposable} over $\FF$ if there exist $h\in\FF[\mathbf{x}]$ and a univariate polynomial $g\in\FF[y]$ such that $\deg(g)\geq 2$ and $f=g(h)$.
Otherwise, $f$ is said to be \emph{indecomposable} over $\FF$.
\end{defn}

Obviously, if a polynomial $f\in\FF_q[\mathbf{x}]$ over a finite field $\FF_q$ is indecomposable over $\overline{\FF}_q$, then it is also indecomposable over $\FF_q$.
The following lemma, which is a special case of \cite[Theorem~4.2]{BDN09}, states that the converse is also true.

\begin{lem}[{\cite[Theorem~4.2]{BDN09}}]\label{lem:base-change}
A polynomial $f\in\FF_q[\mathbf{x}]$ that is indecomposable over $\FF_q$ is also indecomposable over $\overline{\FF}_q$.
\end{lem}

In \cite{DV22}, Derksen and Viola proved the following result, which states that if a polynomial is indecomposable, then its outputs are equidistributed.

\begin{lem}[{\cite[Lemma~12]{DV22}}]\label{lem:equidistributed}
There exists an absolute constant $C>0$ such that the following holds: Suppose $f\in\FF_q[\mathbf{x}]=\FF_q[x_1,\dots,x_n]$ is indecomposable over $\FF_q$. Then $f(\mathbf{U}_{\FF_q^n})$ is $\epsilon$-close to $\mathbf{U}_{\FF_q}$, where $\epsilon=C d^2/\sqrt{q}$.
\end{lem}

The proof of \cref{lem:equidistributed} is based on the observation that the indecomposability of $f$ precisely captures the property that for most $b\in\FF_q$, the variety $f^{-1}(b)$, defined by the constraint $f(\mathbf{x})=b$, is \emph{absolutely irreducible}. This condition of absolute irreducibility is required by the \emph{Weil bound} \cite{Wei49}. Consequently, one can apply the Weil bound to show that for most $b$, the number of points in $f^{-1}(b)\cap \FF_q^n$ is close to $q^{n-1}$, thereby proving the equidistribution of the output of $f$. For details, we refer the reader to \cite{DV22}.

Finally, the following lemma connects indecomposability with irreducibility over algebraically closed fields. It is explicitly stated in, e.g., \cite{CN10}.

\begin{lem}[{\cite[Lemma~7]{CN10}}]\label{lem:indec2irred}
Let $f\in\FF[\mathbf{x}]$ be a non-constant polynomial over a field $\FF$.
Then $f$ is indecomposable over $\overline{\FF}$ iff $f - t$ is irreducible over $\overline{\FF(t)}$, where $t$ is a new variable.
\end{lem}

\paragraph{Resultants.}
Let $f(x) = \sum_{i=0}^{d_1} a_i y^i$ and $g(x) = \sum_{i=0}^{d_2} b_i y^i$ be two univariate polynomials in $y$ over a field $\FF$ and suppose that $d_1+d_2 > 0$. 
The \emph{Sylvester Matrix} of $f$ and $g$ is the $(d_1+d_2) \times (d_1 + d_2)$ matrix

\[
\begin{pmatrix}
	a_{0} & & & &  b_{0} & & & & \\
	a_{1} & a_{0} & &  & b_{1} & b_{0} & & \\
	a_{2} & a_{1} & \ddots & & b_{2} & b_{1} & \ddots & & \\
	\vdots & & \ddots & a_{0} & \vdots & & \ddots & & b_{0} \\
	& \vdots & & a_{1} & b_{d_2} & \vdots & &  & b_{1} \\
	a_{d_1} & & & &  & b_{d_2}  & & & \\
	& a_{d_1} & & \vdots & &  & & \vdots & \\
	& & \ddots & & & & \ddots & & \\
	& & & a_{d_1} & & & &  & b_{d_2}
\end{pmatrix}_{\raisebox{2pt}{.}}
\]

The determinant of this matrix is called the \emph{resultant} of $f$ and $g$, and is denoted $\Res{f,g}$. It holds that $f$ and $g$ have a common factor if and only if $\Res{f,g}=0$ (\cite[Proposition 3 in Chapter 3, Section 6]{CLO07}).

Thus, in the case where $g=\pdv{f}{y}$, it holds that $\Res{f,g} \neq 0$ if and only if $f$ does not have a root of multiplicity greater than one.

\paragraph{Hensel Lifting.}
Hensel lifting is a general technique for ``lifting'' roots or factorizations of a polynomial modulo an ideal $I$ of a ring $R$ to those modulo powers of $I$, under some mild conditions.
The use of Hensel's lifting lemma is standard in multivariate factorization algorithms, and it is available in various forms. We state one standard form, which can be derived from \cite[Theorem~7.3]{Eis95} as a special case. This form is particularly relevant to our discussion of Lecerf's techniques in \cref{sec:lecerf}.

\begin{lem}[Hensel's lifting lemma]\label{lem:hensel}
Let $f\in\FF[x_1,\dots,x_n,y]=\FF[\mathbf{x},y]$ be a nonzero polynomial over a field $\FF$. Suppose $\bar{\lambda}\in \FF$ is a simple root of $f(\mathbf{0},y)\in\FF[y]$.
Then there exists unique $\lambda\in\FF[[\mathbf{x}]]$ such that
\begin{enumerate}
    \item $f(\mathbf{x},\lambda)=0$, i.e., $\lambda$ is a root of $f$ as a univariate polynomial in $y$ over $\FF[\mathbf{x}]$, and
    \item $\lambda(\mathbf{0})=\bar{\lambda}$.
\end{enumerate}
\end{lem}

\section{Hypothesis~(H)}
\label{sec:H}

Lecerf's papers \cite{Lec06, Lec07} on multivariate polynomial factoring assume a hypothesis about the polynomial $f$, which he calls Hypothesis~(H). Such a hypothesis can be satisfied with high probability by applying a random linear transformation on the variables.

In this section, we discuss Lecerf's Hypothesis~(H) and show that, for our purpose, the random linear transformation can be derandomized by using a HSG for polynomials of degree $O(d)$. The fact that we are interested in the irreducibility of $f-t$ for an indeterminate $t$, rather than that of $f$, is crucial in keeping the degree linear in $d$.

Let $\FF$ be a field.
First, we define Hypothesis~(H).

\begin{defn}[{Hypothesis~(H) \cite{Lec06, Lec07}}]\label{defn:hypH}
Let $f\in\FF[x_1,\dots,x_n,y]=\FF[\mathbf{x}, y]$ be a non-constant polynomial. We say $f$ satisfies \emph{Hypothesis~(H)} if
\begin{enumerate}
\item $f$ is monic in $y$ and $\deg_y(f)=\deg(f)$,
\item $\Res{f(\mathbf{0},y),\pdv{f}{y}(\mathbf{0},y)}\neq 0$.
\end{enumerate}
\end{defn}

We also need a family of invertible linear transformations defined as follows.

\begin{defn}
For $\mathbf{a}=(a_1,\dots,a_n)\in\FF^n$, 
let $s_{\mathbf{a}}$ be the $\FF$-linear automorphism of $\FF[\mathbf{x}, y]$ that fixes $y$ and sends $x_i$ to $x_i+a_i y$.
\end{defn}

\begin{lem}\label{lem:hypH-1}
Let $f\in \FF[\mathbf{x}, y]$ be a nonzero polynomial of degree at most $d$. Then there exists a nonzero polynomial $B\in \FF[\mathbf{x}]$ of degree at most $d$ such that for every $\mathbf{a}\in\FF^n$ satisfying $B(\mathbf{a})\neq 0$, it holds that $\deg_y (s_{\mathbf{a}}(f))=d$ and the coefficient of $y^d$ in $s_{\mathbf{a}}(f)$ is in $\FF^\times$.
\end{lem}

\begin{proof}
Let $f_d$ be the degree-$d$ homogeneous part of $f$, so that we can write $f=f_d+g$ where $g=f-f_d$ has degree less than $d$.
Write $f_d=\sum_{i=0}^d c_i(\mathbf{x}) y^i$, where each $c_i\in\FF[\mathbf{x}]$ is either zero or a homogeneous polynomial of degree $d-i$.

Consider $\mathbf{a}\in\FF^n$. Note that
\[
s_{\mathbf{a}}(f)=s_{\mathbf{a}}(f_d)+s_{\mathbf{a}}(g)=\sum_{i=0}^d s_{\mathbf{a}}(c_i(\mathbf{x})) y^i + s_{\mathbf{a}}(g)
=\sum_{i=0}^d c_i(\mathbf{x}+y\cdot \mathbf{a}) y^i + g(\mathbf{x}+y\cdot \mathbf{a}, y).
\]
As $\deg(g)\leq d$ and each $c_i$ is either zero or homogeneous of degree $d-i$, we have that $\deg_y s_{\mathbf{a}}(f)\leq d$, and that the coefficient of $y^d$ in $s_{\mathbf{a}}(f)$ is $\sum_{i=0}^d c_i(\mathbf{a})\in\FF$.
So we may choose $B=\sum_{i=0}^d c_i$, which is a nonzero polynomial of degree at most $d$.
\end{proof}

\begin{lem}\label{lem:hypH-2}
Assume $f\in \FF[\mathbf{x}, y]$ is a polynomial of degree $d\geq 1$ that satisfies Item~1 of Hypothesis~(H).
Further assume that $\mathrm{char}(\FF)$ is either zero or greater than $d$. Let $c\in\FF^\times$. Then $f+ct$ is a degree-$d$ polynomial satisfying Hypothesis~(H) as a polynomial over $\FF(t)$.
\end{lem}

\begin{proof}
As $f$ satisfies Item~1 of Hypothesis (H) and has degree $d\geq 1$, so does $f+ct$.
So it suffices to verify Item~2.
Write $f=\sum_{i=0}^d c_i y^i$, where $c_i\in \FF[\mathbf{x}]$ and $c_d=1$.
Then $\pdv{(f+ct)}{y}=\sum_{i=1}^{d} (i\cdot c_i) y^{i-1}$, which has degree $d-1$ in $y$ since $d \cdot c_{d}=d\neq 0$ by the assumption about $\mathrm{char}(\FF)$. 

Let $\bar{c}_i=c_i(\mathbf{0})$ for $i=0,1,\dots,d$.
Let $h=\Res{(f+ct)(\mathbf{0},y),\pdv{f+ct}{y}(\mathbf{0},y)}$.
Then $h$ is the determinant of the following $(2d-1)\times (2d-1)$ matrix:
\[
\begin{pmatrix}
\bar{c}_0 + ct & 0 & \cdots  & 0 & \bar{c}_1 & 0 & \cdots & 0\\
\bar{c}_1 & \bar{c}_0 + ct & \cdots & 0 & 2\bar{c}_2 & \bar{c}_1 & \cdots & 0\\
\bar{c}_2 & \bar{c}_1 & \ddots & 0 & 3\bar{c}_3 & 2\bar{c}_2 & \ddots & 0\\
\vdots & \vdots & \ddots & \bar{c}_0 + ct & \vdots & \vdots & \ddots & \bar{c}_1 \\
\bar{c}_d & \bar{c}_{d-1} & \cdots & \vdots & d \bar{c}_d & (d-1) \bar{c}_{d-1} & \cdots & \vdots\\
0 & \bar{c}_d & \ddots & \vdots & 0 & d \bar{c}_d & \ddots & \vdots\\
\vdots & \vdots & \ddots & \bar{c}_{d-1} & \vdots & \vdots & \ddots & (d-1) \bar{c}_{d-1}\\
0 & 0 & \cdots & \bar{c}_d & 0 & 0 & \cdots & d \bar{c}_d
\end{pmatrix}_{\raisebox{2pt}{.}}
\]

Observe that $\deg_t h \leq d-1$, and that the coefficient of $t^{d-1}$ in $h$ is $c^{d-1}(d\bar{c}_{d})^d=c^{d-1}d^d\neq 0$ since only those entries on the diagonal contribute to this coefficient. This implies that $h\neq 0$, i.e., $f+ct$ satisfies Item~2 of Hypothesis (H).
\end{proof}

\begin{cor}\label{cor:goods}
Assume that $f\in \FF[\mathbf{x}, y]$ is a polynomial of degree $d\geq 1$ and that $\mathrm{char}(\FF)$ is either zero or greater than $d$.
Then there exists a nonzero polynomial $B\in \FF[\mathbf{x}]$ of degree at most $d$ such that for every $\mathbf{a}\in\FF^n$ satisfying $B(\mathbf{a})\neq 0$, $s_{\mathbf{a}}(f)-t$ equals a product $c\cdot g$ where $c\in\FF^\times$ and $g\in\FF(t)[\mathbf{x}, y]$ is a degree-$d$ polynomial satisfying Hypothesis~(H).
\end{cor}

\begin{proof}
Let $B$ be as in \cref{lem:hypH-1}. Consider $\mathbf{a}\in\FF^n$ satisfying $B(\mathbf{a})\neq 0$. By \cref{lem:hypH-1}, we may write $s_{\mathbf{a}}(f)=c\cdot\tilde{g}$ where $c\in\FF^\times$ and $\tilde{g}$ satisfies Item~1 of Hypothesis~(H).
Then $s_{\mathbf{a}}(f)-t=c\cdot\tilde{g}-t=c \cdot g$ where $g=\tilde{g}-c^{-1} t$.
By \cref{lem:hypH-2}, $g$ is a degree-$d$ polynomial satisfying Hypothesis~(H).
\end{proof}

Thus, by choosing good $\mathbf{a}\in\FF$ via an explicit HSG for polynomials of degree at most $d$ and performing the transformation $f\mapsto s_{\mathbf{a}}(f)$, we may assume $f-t$ satisfies Hypothesis (H).

\paragraph{Satisfying Hypothesis (H) in Small Characteristics.} 

While our final result needs $\mathrm{char}(\FF)>d(d-1)$, the assumption that $\mathrm{char}(\FF)$ is zero or large enough is not crucial for the sake of satisfying Hypothesis (H).
We now sketch how to modify the proof of \cref{lem:hypH-2} when $0<\mathrm{char}(\FF)\leq d$.

Let $p=\mathrm{char}(\FF)>0$. For our purpose, we may assume $\FF$ is a perfect field and $f$ is indecomposable over $\FF$. This implies that $f\not\in\FF[x_1^p,\dots,x_n^p,y^p]$. 
Then it is not hard to show that there exists an integer $e>0$ coprime to $p$ such that for random $\mathbf{a}\in\FF^n$, with high probability, not only is the coefficient of $y^d$ in $s_{\mathbf{a}}(f)$ nonzero, but so is the coefficient of $y^e$. Choose the largest $e$ that has this property.
After replacing $f$ by $s_{\mathbf{a}}(f)$, the polynomial $\pdv{f+ct}{y}$ in the proof of \cref{lem:hypH-2} would have degree $e-1$ instead of $d-1$ in $y$.
Then $\deg_t(h)=e-1$ and the coefficient of $t^{e-1}$ in $h$ is $c^{e-1}(e \bar{c}_e)^d$, which is nonzero iff $\bar{c}_e=c_e(\mathbf{0})$ is nonzero. The latter condition can be guaranteed with high probability by performing the substitutions $x_i\mapsto x_i+b_i$ for random $\mathbf{b}=(b_1,\dots,b_n)\in\FF^n$. Finally, it is not difficult to show that the choices of $\mathbf{a}$ and $\mathbf{b}$ can be derandomized by using an explicit HSG for polynomials of degree $O(d)$.

\section{Lecerf's Techniques}
\label{sec:lecerf}

We describe Lecerf's techniques in this section. For simplicity, our discussion is restricted to the special case where the base field is algebraically closed.

Let $\KK$ be an algebraically closed field, and let $f\in\KK[\mathbf{x}, y]$ be a polynomial of degree $d\geq 1$ satisfying Hypothesis~(H).
Define $\bar{f}:=f(\mathbf{0}, y)\in \KK[y]$.
As $\KK$ is algebraically closed and $\Res{f(\mathbf{0},y),\pdv{f}{y}(\mathbf{0},y)}\neq 0$, the univariate polynomial $\bar{f}$ factorizes into distinct linear factors
\[
\bar{f}=\prod_{i=1}^d  (y -\bar{\lambda}_i) 
\]
where $\bar{\lambda}_i\in\KK$ for $i\in [d]$.
By Hensel's lifting lemma, the above factorization of $\bar{f}$ over $\KK$ lifts to a factorization of $f$ into distinct linear factors
\[
f=\prod_{i=1}^d (y-\lambda_i),
\] 
where $\lambda_i\in \KK[[\mathbf{x}]]$ and $\lambda_i(\mathbf{0})=\bar{\lambda}_i$ for $i\in [d]$.

Now we introduce new variables $\mathbf{z}=(z_1,\dots,z_n)$ and $x$, and define $g:=f(z_1 x ,\dots, z_n x, y)\in \KK[\mathbf{z},x,y]$. Then $g$ factorizes into linear factors
\[
g=\prod_{i=1}^d (y-\lambda_i(z_1 x, \dots, z_n x))
\]
where each $\lambda_i(z_1 x, \dots, z_n x)$ lives in $\KK[\mathbf{z}][[x]]$.
For $i\in [d]$, let $g_i$ be the factor $y-\lambda_i(z_1 x, \dots, z_n x)$ of $g$, and let $\hat{g}_i$ be its cofactor $\prod_{j\in [d]\setminus\{i\}} g_j$. So $g_i,\hat{g}_i\in \KK[\mathbf{z}][[x]][y]$. 

For $h\in A[[x]][y]$ over a commutative ring $A$ and $(j,k)\in\mathbb{N}^2$, denote by $\coeff{h, x^j y^k}\in A$ the coefficient of $x^j y^k$ in $h$.
We are now ready to define the linear system $D_{\mathbf{z},\sigma}$ used in \cite{Lec06, Lec07}.

\begin{defn}[{Linear system $D_{\mathbf{z},\sigma}$ \cite{Lec06, Lec07}}]\label{defn:system}
Let $\sigma\in\NN$.
Define $D_{\mathbf{z},\sigma}$ to be the following linear system over $\KK(\mathbf{z})$ in the unknowns $\ell_1,\dots,\ell_d$:
\[
D_{\mathbf{z},\sigma}
\begin{dcases}
\sum_{i=1}^d \coeff{\hat{g}_i \pdv{g_i}{y}, x^j y^k}\cdot \ell_i=0,\quad k\leq d-1,~d\leq j+k\leq \sigma-1,\\
\sum_{i=1}^d \coeff{\hat{g}_i \pdv{g_i}{x}, x^j y^k}\cdot \ell_i=0,\quad k\leq d-1,~j\leq \sigma-2,~d\leq j+k\leq \sigma-1.
\end{dcases}
\]
\end{defn}

We have the following easy lemma.
\begin{lem}\label{lem:degree-bound}
For $(j,k)\in\NN^2$,
$\coeff{\hat{g}_i \pdv{g_i}{x}, x^j y^k},\coeff{\hat{g}_i \pdv{g_i}{y}, x^j y^k}\in\KK[\mathbf{z}]$ are polynomials of degree at most $j+1$ and $j$ respectively.
\end{lem}

\begin{proof}
Consider arbitrary $i\in [d]$ and $(j,k)\in\NN^2$.
As $g_i=y-\lambda_i(z_1 x, \dots, z_n x)$ and $j\geq 0$, only terms of degree at most $j$ in $z_1,\dots,z_n$ contribute to the coefficient of $x^j y^k$ in $g_i$. 
Then, as the operator $\pdv{}{x}$ is linear and sends $x^u y^{v}$ to $u x^{u-1} y^{v}$ for all $u, v\in\NN$, one can see that only terms of degree at most $j+1$ in $z_1,\dots,z_n$ contribute to the coefficient of $x^j y^k$ in $\pdv{g_i}{x}$. 
Also, $\pdv{g_i}{y}=1$ by definition.

For a collection of polynomials $h_1, \dots, h_s\in \KK[\mathbf{z}][[x]][y]$ and $h=\prod_{i=1}^s h_i$, we have 
\begin{equation}\label{eq:coeff-product}
\coeff{h,  x^j y^k}=\sum_{\substack{j_1,\dots,j_s,k_1,\dots,k_s\in \NN \\ \sum_i j_i=j, \sum_i k_i=k}}
\prod_{i=1}^s \coeff{h_i, x^{j_i}y^{k_i}}.
\end{equation}
We already know $\deg\left(\coeff{g_i, x^j y^k}\right)\leq j$, $\deg\left(\coeff{\pdv{g_i}{x}, x^j y^k}\right)\leq j+1$, and $ \deg\left(\coeff{\pdv{g_i}{y}, x^j y^k}\right)=0$ for $i\in[d]$ and $(j,k)\in\NN^2$ by the above discussion.
Choosing $(h_1,\dots,h_s)$ to be $(g_1, \dots, g_{i-1}, g_{i+1},\dots g_d, \pdv{g_i}{x})$ and $(g_1, \dots, g_{i-1}, g_{i+1},\dots g_d, \pdv{g_i}{y})$ respectively and applying \eqref{eq:coeff-product} proves the claim.
\end{proof}

For $\mathbf{a}=(a_1,\dots,a_n)\in\KK^n$, we can assign $a_1\dots,a_n$ to $z_1,\dots,z_n$ respectively in the polynomials $\coeff{\hat{g}_i \pdv{g_i}{x}, x^j y^k},\coeff{\hat{g}_i \pdv{g_i}{y}, x^j y^k}\in\KK[\mathbf{z}]$. This yields a linear system over $\KK$, called the \emph{specialization} of $D_{\mathbf{z},\sigma}$ at $\mathbf{a}$ and denoted by $D_{\mathbf{a},\sigma}$.

\begin{defn}[Specialization]
For $\sigma\in\NN$ and $\mathbf{a}=(a_1,\dots,a_n)\in\KK^n$, define $D_{\mathbf{a},\sigma}$ to be the following linear system over $\KK$ in the unknowns $\ell_1,\dots,\ell_d$:
\[
D_{\mathbf{a},\sigma}
\begin{dcases}
\sum_{i=1}^d \coeff{\hat{g}_i \pdv{g_i}{y}, x^j y^k}(\mathbf{a})\cdot \ell_i=0,\quad k\leq d-1,~d\leq j+k\leq \sigma-1,\\
\sum_{i=1}^d \coeff{\hat{g}_i \pdv{g_i}{x}, x^j y^k}(\mathbf{a})\cdot \ell_i=0,\quad k\leq d-1,~j\leq \sigma-2,~d\leq j+k\leq \sigma-1.
\end{dcases}
\]
\end{defn}

For $S\subseteq [d]$, define $\delta_S=(\delta_{S,1},\dots,\delta_{S,d})\in \KK^d$ by
\[
\delta_{S,i}=\begin{cases}
1 & i\in S,\\
0 & i\not\in S.
\end{cases}
\]

For every factor $\tilde{f}$ of $f$, we may associate a set $S\subseteq [d]$ such that $\tilde{f}=\prod_{i\in S} (y-\lambda_i)$, i.e., $S$ is the set of indices $i\in [d]$ such that $y-\lambda_i$ divides $\tilde{f}$.
Then irreducible factors $f_1,\dots,f_r$ of $f$ over $\KK$ are then associated with sets $S_1,\dots,S_r\subseteq [d]$, which form a partition of $[d]$.
In \cite{Lec06, Lec07}, Lecerf proved that, when $\sigma$ is large enough, the solution space of $D_{\mathbf{z},\sigma}$ is exactly spanned by the vectors $\delta_{S_1},\dots,\delta_{S_r}$, and a similar statement holds for the specializations $D_{\mathbf{a},\sigma}$. We state Lecerf's results formally as the following theorem.

\begin{thm}[{\cite{Lec06, Lec07}}]\label{thm:lecerf}
Assume $\mathrm{char}(\KK)$ is zero or greater than $d(d-1)$. Let $\sigma\geq 2d$.
Let $f\in\KK[\mathbf{x}, y]$ be a polynomial of degree $d\geq 1$ satisfying Hypothesis~(H). Then:
\begin{enumerate}
\item Suppose $f=\prod_{i=1}^r f_i$ is the factorization of $f$ into its irreducible factors over $\KK$. For $i\in [r]$, let $S_i$ be the set of indices $j\in [d]$ such that $y-\lambda_j$ divides $f_i$. Then $\delta_{S_1},\dots,\delta_{S_r}$ form a basis of the solution space of $D_{\mathbf{z},\sigma}$.
\item Let $\mathbf{a}=(a_1,\dots,a_n)\in\KK^n$ and $f_{\mathbf{a}}=f(a_1 x, \dots,a_n x, y)\in\KK[x,y]$.
Suppose $f_{\mathbf{a}}=\prod_{i=1}^s f_{\mathbf{a}, i}$ is the factorization of $f_{\mathbf{a}}$ into its irreducible factors over $\KK$. For $i\in [s]$, let $S_{\mathbf{a}, i}$ be the set of indices $j\in [d]$ such that $y-\lambda_j(a_1 x, \dots,a_n x)$ divides $f_{\mathbf{a},i}$. Then $\delta_{S_{\mathbf{a},1}},\dots,\delta_{S_{\mathbf{a},s}}$ form a basis of the solution space of $D_{\mathbf{a},\sigma}$.
\end{enumerate}
\end{thm}

The first item of \cref{thm:lecerf} is explicitly stated as \cite[Lemma~1]{Lec07}. The second item follows from \cite[Theorem~1 and Lemma~4]{Lec06}. For a detailed analysis of the linear systems $D_{\mathbf{z},\sigma}$ and $D_{\mathbf{a},\sigma}$, and for a conceptual interpretation of these linear systems in terms of the closedness condition of differential 1-forms, we refer the reader to \cite{Lec06, Lec07}, as well as to the earlier paper of Gao \cite{Gao2003}.

\paragraph{Bertinian Good/Bad Points.}

The classical Bertini irreducibility theorem \cite{Sha94} states, among other things, that over an algebraically closed field $\KK$, the intersection of an irreducible variety with a plane in general position is still irreducible.
This motivates the following definition:

\begin{defn}[Bertinian good/bad points \cite{Lec07}]
Let $f\in\KK[\mathbf{x}, y]$ be a non-constant polynomial satisfying Hypothesis~(H). We say $\mathbf{a}=(a_1,\dots,a_n)\in\KK^n$ is a \emph{Bertinian good point} for $f$ if for every irreducible factor $\tilde{f}$ of $f$ over $\KK$, the bivariate polynomial $\tilde{f}_{\mathbf{a}}=\tilde{f}(a_1 x,\dots,a_n x, y)$ is also irreducible over $\KK$.
We say $\mathbf{a}=(a_1,\dots,a_n)\in\KK^n$ is a \emph{Bertinian bad point} for $f$ if it is not a Bertinian good point for $f$.  
\end{defn}

Lecerf \cite[Theorem~6]{Lec07} proved that given $f$, there exists a nonzero polynomial $Q\in \KK[z_1,\dots,z_n]$ of degree at most $(d-1)(2d-1)$ that vanishes at all Bertinian bad points for $f$, where $d=\deg(f)$. 
Let $M$ be the matrix representing the linear system $D_{\mathbf{z},\sigma}$.
Lecerf's proof can be sketched as follows: By \cref{thm:lecerf}, the solution space of $D_{\mathbf{a},\sigma}$ contains that of $D_{\mathbf{z},\sigma}$, and $\mathbf{a}$ is Bertinian good as long as the two are equal. Thus, we may choose $Q$ to be the determinant of the largest nonsingular submatrix of $M$. This is because for such $Q$, if $Q$ does not vanish at $\mathbf{a}$, then $D_{\mathbf{z},\sigma}$ and $D_{\mathbf{a},\sigma}$ have the same rank, and hence their solution spaces must be equal. The bound $(d-1)(2d-1)$ on the degree of $Q$ follows from \cref{lem:degree-bound}. 

In \cite{Lec07}, Lecerf also demonstrated that the degree bound $(d-1)(2d-1)$ is asymptotically tight by providing an example for which a degree of $\Omega(d^2)$ of the polynomial $Q$ is necessary. However, our next lemma states that, perhaps surprisingly, the degree bound can be improved to $2d-1$ if we allow the use of the zero loci of \emph{multiple} polynomials to cover the Bertinian bad points for $f$. For simplicity, we state the lemma in the special case where $f$ is irreducible, which suffices for our purpose.

\begin{lem}\label{lem:Bertinian}
Assume $\mathrm{char}(\KK)$ is zero or greater than $d(d-1)$.
Let $f\in\KK[\mathbf{x}, y]$ be a irreducible polynomial over $\KK$ of degree $d\geq 1$ satisfying Hypothesis~(H). Let $m=2^{d-1}-1$. Then there exist nonzero polynomials $Q_1,\dots,Q_m\in \KK[\mathbf{z}]=\KK[z_1,\dots,z_n]$ of degree at most $2d-1$ such that for every Bertinian bad point $\mathbf{a}\in\KK^n$ for $f$, at least one polynomial $Q_i$ vanishes at $\mathbf{a}$.
\end{lem}

\begin{proof}
Let $\sigma=2d$.
Let $N$ be the number of equations in $D_{\mathbf{z},\sigma}$. Let $M$ be the $N\times d$ matrix over $\FF(\mathbf{z})$ representing the linear system $D_{\mathbf{z},\sigma}$.
Note that by \cref{defn:system}, the entries of $M$ are of the form $\coeff{\hat{g}_i \pdv{g_i}{x}, x^j y^k}$ with $j\leq \sigma-2$ or $\coeff{\hat{g}_i \pdv{g_i}{y}, x^j y^k}$ with $j\leq \sigma-1$.
By \cref{lem:degree-bound} and the fact that $\sigma=2d$, the entries of $M$ are polynomials in $\KK[\mathbf{z}]$ of degree at most $2d-1$.

There are exactly $m=2^{d-1}-1$ proper subsets of $[d]$ containing $1$. Let $S_1,\dots,S_m$ be an enumeration of them.
Consider $i\in [m]$. As $f$ is irreducible, by \cref{thm:lecerf}, $\delta_{S_i}$ is not in the solution space of $D_{\mathbf{z},\sigma}$.
So we can fix a row $\mathbf{r}_i=(r_{i,1},\dots,r_{i,d})$ of $M$ such that the inner product of $\mathbf{r}_i$ and $\delta_{S_i}$ is nonzero, i.e., $\sum_{j=1}^d r_{i,j} \delta_{S_i, j}\neq 0$. Let $Q_i=\sum_{j=1}^d r_{i,j} \delta_{S_i, j}$, which is a nonzero polynomial in $\KK[\mathbf{z}]$ of degree at most $2d-1$. Choose $Q_i$ in this way for each $i=1,\dots, m$. 

Now let $\mathbf{a}$ be a Bertinian bad point for $f$. Then $f_{\mathbf{a}}=f(a_1 x, \dots,a_n x, y)$ factorizes into more than one irreducible factor over $\KK$. Let $\tilde{f}_{\mathbf{a}}$ be the irreducible factor of $f_{\mathbf{a}}$ divisible by $y-\lambda_1(a_1 x, \dots,a_n x)$. Let $S$ be the set of $j\in [d]$ such that $\tilde{f}_{\mathbf{a}}$ is divisible by $y-\lambda_j(a_1 x, \dots,a_n x)$. Then $S$ is a proper subset of $[d]$ containing $1$. So $S=S_i$ for some $i\in [m]$. By \cref{thm:lecerf}, $\delta_{S_i}$ is in the solution space of $D_{\mathbf{a},\sigma}$. As $D_{\mathbf{a},\sigma}$ is the specialization of $D_{\mathbf{z},\sigma}$ at $\mathbf{a}$, the vector $(r_{i,1}(\mathbf{a}),\dots,r_{i,d}(\mathbf{a}))$ is a row of the matrix representing $D_{\mathbf{a},\sigma}$. So $\sum_{j=1}^d r_{i,j}(\mathbf{a}) \delta_{S_i, j}=0$, i.e., $Q_i(\mathbf{a})=0$.
\end{proof}

\paragraph{The Number of Low-Degree Polynomials Needed.}

It is an intriguing mathematical question to us how many low-degree polynomials are needed to cover the Bertinian bad points for $f$. We now formalize this question. 

\begin{defn}\label{defn:N}
Let $\KK$ be an algebraically closed field.
For positive integers $d$ and $D$, define $N(d, D, \KK)$ to be the smallest $N\in\NN$ such that the following holds:
Let $f\in\KK[\mathbf{x}, y]$ be an irreducible polynomial of degree at most $d$ over $\KK$ satisfying Hypothesis~(H). 
Then there exist $N$ nonzero polynomials in $\KK[\mathbf{z}]$ of degree at most $D$ such that the union of the zero loci of these polynomials contains all Bertinian bad points for $f$ in $\KK^n$. 

If such $N$ does not exist, define $N(d,D,\KK)=\infty$.
\end{defn}

In our application, it suffices to consider polynomials of the special form $f+c\cdot t$, where $c\in\FF^\times$, $f\in\FF[\mathbf{x},y]$ and $\KK=\overline{\FF(t)}$. Moreover, by performing a variable substitution $t\mapsto -c^{-1} t$, we may assume $c=-1$. This motivates us to introduce the following variant of \cref{defn:N}:
 
\begin{defn}
Let $\FF$ be a field.
For positive integers $d$ and $D$, define $N^*(d, D, \FF)$ to be the smallest $N\in\NN$ such that the following holds:
Let $f\in\FF[\mathbf{x}, y]$ be a polynomial of degree at most $d$ such that $f-t$ is an irreducible polynomial over $\overline{\FF(t)}$ satisfying Hypothesis~(H). 
Then there exist $N$ nonzero polynomials in $\overline{\FF(t)}[\mathbf{z}]$ of degree at most $D$ such that the union of the zero loci of these polynomials contains all Bertinian bad points for $f-t$ in $\overline{\FF}^n$. 

If such $N$ does not exist, define $N^*(d, D, \FF)=\infty$.
\end{defn}

Lecerf's result \cite[Theorem~6]{Lec07} can be interpreted as the statement that when $\mathrm{char}(\KK)$ is zero or greater than $d(d-1)$, it holds that 
\[
N(d, D, \KK)=1 \text{ for } D\geq (d-1)(2d-1).
\]
Our \cref{lem:Bertinian} states that under the same condition, we have 
\[
N(d, D, \KK)\leq 2^{d-1}-1 \text{ for } D\geq 2d-1.
\]
In \cite{Lec07}, Lecerf gave an example showing that the degree bound $O(d^2)$ for the smallest $D$ satisfying $N(d, D, \KK)=1$ is asymptotically tight.\footnote{See the example before Theorem~6 in \cite{Lec07}.} As one can always combine the $N(d, D, \KK)$ polynomials of degree at most $D$ into a single polynomial of degree at most $N(d, D, \KK)\cdot D$ by taking their product, this implies $N(d, D, \KK)\cdot D=\Omega(d^2)$, i.e., $N(d, D, \KK)=\Omega(d^2/D)$. 

\begin{question}
\label{question:N}
Give improved upper bounds (or lower bounds) on $N(d, D, \KK)$ and $N^*(d, D, \FF)$, at least when the characteristic of $\KK$ or $\FF$ is zero or large enough.
\end{question}

By definition, $N^*(d, D, \FF)\leq N(d, D, \overline{\FF(t)})$. A subexponential upper bound on $N^*(d, D, \FF)$ for $D=O(d)$ will improve the required field size in \cref{thm:main}.

Finally, it might be possible to exploit some extra structure to derive better bounds on $N^*(d, D, \FF)$ than those obtained for $N(d, D, \KK)$.
For example, if we modify the definition of $N^*(d, D, \FF)$ by only considering those polynomials $f$ of \emph{prime} degree, then $N^*(d, 2d-1, \FF)\leq 1$. This is because if $f_{\mathbf{a}}-t$ is reducible over $\overline{\FF(t)}$, then $f_{\mathbf{a}}$ is decomposable over $\overline{\FF}$ by \cref{lem:indec2irred}. 
But as $\deg(f)$ is prime, $f_{\mathbf{a}}$ must be of the form $g(h)$ with $\deg(g)=\deg(f)$ and $\deg h=1$. This in turn implies that $f_{\mathbf{a}}-t$ factorizes into $\deg(f)$ linear factors over $\overline{\FF_q(t)}$. 
In \cref{thm:prime-degree}, we use this idea to show that the required field size can be improved to $O(d^4/\eps^2)$ if we only want to fool polynomials whose degrees are prime and at most $d$.

\section{Proofs of the Main Theorems} 

In this section, we present our PRG construction and prove the main theorems. 

Let $n$ and $d$ be positive integers. Let $\FF_q$ be a finite field of characteristic at least $d(d-1)+1$.
We now present the construction of our PRG 
\[
G: S\to \FF_q^{n+1}
\]
for polynomials $f\in\FF_q[\mathbf{x},y]=\FF_q[x_1,\dots,x_n,y]$ of degree at most $d$.
To simplify our notation, these polynomials are assumed to be $(n+1)$-variate rather than $n$-variate.

\begin{construction}\label{construction:prg}
The construction is as follows:

\begin{itemize}
\item Let $H: T\to \FF_q^{n}$ be an explicit HSG for $n$-variate polynomials of degree at most $2d-1$ over $\FF_q$ with density $1-\delta$ and seed length $\log |T|=O(d\log n+\log(1/\delta))$, where $\delta=C_0 (2d-1)/q$ and $C_0>0$ is an absolute constant.  For $i\in [n]$ and $s\in T$, denote the $i$-th coordinate of $H(s)$ by $H(s)_i$.
The existence of $H$ is guaranteed by \cref{thm:optimalHSG}.
\item Let $S=T\times T\times \FF_q\times \FF_q$. Define $G:S\to \FF_q^{n+1}$ by
\[
G(r, s, u, v)=(H(s)_1\cdot u+ H(r)_1\cdot v, \dots, H(s)_n\cdot u+H(r)_n\cdot v, v). 
\]
\end{itemize}
\end{construction}

In other words, we use random $(r,s)\in T\times T$ to pick a plane in $\FF_q^{n+1}$, and use random $(u,v)\in\FF_q\times \FF_q$ to pick a point on the plane. The following lemma states that with high probability, a given indecomposable polynomial $f\in\FF_q[\mathbf{x},y]$ remains indecomposable when restricted to the plane.

\begin{lem}
\label{lem:reduction}
Let $f\in \FF_q[\mathbf{x},y]$ be an indecomposable polynomial of degree at most $d$ over $\FF_q$. 
Let $(r,s)$ be a random element of $T\times T$.
Let $\mathbf{a}=(a_1,\dots,a_n)=H(r)$ and $\mathbf{b}=(b_1,\dots,b_n)=H(s)$.
Finally, let $F=f(b_1 x+a_1y,\dots,b_n x+a_n y, y)\in\FF_q[x,y]$. Then
\[
\Pr\left[F~\text{is indecomposable over}~\FF_q\right]\geq 1-2^{d-1}\delta.
\]
\end{lem}

\begin{proof}
Recall that $s_{\mathbf{a}}$ is the $\FF_q$-linear automorphism of $\FF_q[\mathbf{x}, y]$ that fixes $y$ and sends $x_i$ to $x_i+a_i y$.
As $f$ is indecomposable over $\FF_q$, so is $s_{\mathbf{a}}(f)$.
By \cref{lem:base-change}, $s_{\mathbf{a}}(f)$ is also indecomposable over $\overline{\FF}_q$.
So $s_{\mathbf{a}}(f)-t$ is irreducible over $\overline{\FF_q(t)}$ by \cref{lem:indec2irred}.

By \cref{cor:goods}, there exists a nonzero polynomial $B\in \FF_q[\mathbf{x}]$ of degree at most $d$ such that if $B(\mathbf{a})\neq 0$, then 
\begin{equation}\label{eq:cdotg}
s_{\mathbf{a}}(f)-t=c\cdot g
\end{equation}
where $c\in\FF_q^\times$ and $g\in\FF_q(t)[\mathbf{x}, y]\subseteq \overline{\FF_q(t)}[\mathbf{x}, y]$ is a degree-$d$ polynomial satisfying Hypothesis~(H). By the HSG property of $H$, the event $B(\mathbf{a})\neq 0$ happens with probability at least $1-\delta$. Condition on this event, so that \eqref{eq:cdotg} holds.
As $s_{\mathbf{a}}(f)-t$ is irreducible over $\overline{\FF_q(t)}$, so is $g$.

Let $m=2^{d-1}-1$.
By \cref{lem:Bertinian}, there exist nonzero polynomials $Q_1,\dots,Q_m\in \overline{\FF_q(t)}[z_1,\dots,z_n]$ of degree at most $2d-1$ such that the union of the zero loci of these polynomials contains all $\mathbf{b}^*=(b_1^*,\dots,b_n^*)\in\overline{\FF}_q^n$ for which
$g(b_1^* x, \dots, b_n^* x, y)$ is reducible over $\overline{\FF_q(t)}$.
By \cref{fact:extension}, $H$ is an HSG with density $1-\delta$ for polynomials of degree at most $2d-1$ over $\overline{\FF_q(t)}$.\footnote{Note that we are applying \cref{fact:extension} to the infinite extension $\overline{\FF_q(t)}/\FF_q$. In principle, it should be possible to make the argument finitary by making some adaptations, such as considering specific values of $t$. However, this may increase the complexity of the proof.}
Therefore, for each $i\in [m]$, the probability that $Q_i(\mathbf{b})=0$ is at most $\delta$.
Condition on the event $Q_1(\mathbf{b}),\dots,Q_m(\mathbf{b})\neq 0$. Then $g(b_1 x, \dots, b_n x,y)$ is irreducible over $\overline{\FF_q(t)}$.
On the other hand, note that
\[
c\cdot g(b_1 x, \dots, b_n x,y)
\stackrel{\eqref{eq:cdotg}}{=}(s_{\mathbf{a}}(f))(b_1 x,\dots,b_n x,y)-t
=f(b_1 x+a_1 y,\dots,b_n x+a_n y,y)-t=F-t
\]
where the second step uses the definition
$s_{\mathbf{a}}(f)=f(x_1+a_1 y,\dots,x_n+a_n y, y)\in\FF_q[\mathbf{x},y]$. So $F-t$ is irreducible over $\overline{\FF_q(t)}$.
By \cref{lem:indec2irred}, $F$ is indecomposable over $\overline{\FF}_q$. So it is indecomposable over $\FF_q$.

The indecomposability of $F$ over $\FF_q$ relies on the conditions $B(\mathbf{a})\neq 0$ and $Q_1(\mathbf{b}),\dots,Q_m(\mathbf{b})\neq 0$. By the union bound, these conditions are simultaneously satisfied with probability at least $1-\delta-m\delta=1-2^{d-1}\delta$, which completes the proof.
\end{proof}

Now we are ready to prove \cref{thm:main}.

\begin{thm}[\cref{thm:main} restated]\label{thm:main2}
There exists an absolute constant $C>0$ such that for $\eps>0$ and $q\geq C(d2^d/\eps+d^4/\eps^2)$ with $\mathrm{char}(\FF_q)\geq d(d-1)+1$,
$G$ as in \cref{construction:prg} is a PRG for $(n+1)$-variate polynomials of degree at most $d$ over $\mathbb{F}_q$ with error $\eps$ and seed length $O(d\log n+\log q)$.
\end{thm}

\begin{proof} 
Let $f\in\FF_q[\mathbf{x}, y]$ be a polynomial of degree at most $d$.
We want to prove that $f(G(\mathbf{U}_S))$ and $f(\mathbf{U}_{\FF_q^{n+1}})$ are $\epsilon$-close (in statistical distance).
We may assume that $f$ is a non-constant polynomial, i.e., $\deg(f)\geq 1$, since the claim is trivial otherwise.

Our next step is the same as in \cite{DV22}: $f$ can always be written in the form $f=g(h)$, where $g\in\FF_q[z]$ is a univariate polynomial and $h\in\FF_q[\mathbf{x},y]$ is indecomposable over $\FF_q$. Let $D=h(G(\mathbf{U}_S))$ and $D'=h(\mathbf{U}_{\FF_q^{n+1}})$. Then $f(G(\mathbf{U}_S))=g(D)$ and $f(\mathbf{U}_{\FF_q^{n+1}})=g(D')$. If $D$ and $D'$ are $\eps$-close, then $g(D)$ and $g(D')$ are also $\eps$-close. Thus, by replacing $f$ with $h$, we may assume that $f$ is indecomposable over $\FF_q$.

Let $r,s,\mathbf{a},\mathbf{b}$ and $F$ be as in \cref{lem:reduction}. Then by \cref{lem:reduction}, the probability that $F$ is decomposable over $\FF_q$ over random $r$ and $s$ is at most $2^{d-1}\delta=C_0 2^{d-1}(2d-1)/q$, where $C_0$ is as in \cref{construction:prg}.
Fix $r$ and $s$ such that $F$ is indecomposable over $\FF_q$.
Then $f(G(r,s,u,v))=F(u,v)$ by definition.
Applying \cref{lem:equidistributed} to $F$ shows that, for such fixed $r$ and $s$, the distribution of $F(u,v)$, i.e., $f(G(r,s,u,v))$, over random $u,v\in\FF_q$ is $\eps'$-close to $\mathbf{U}_{\FF_q}$, where $\eps'=C_1 d^2/\sqrt{q}$ and $C_1>0$ is an absolute constant. 
It follows that the statistical distance between $f(G(\mathbf{U}_S))$ and $\mathbf{U}_{\FF_q}$ is at most $2^{d-1}\delta+\eps'$.

On the other hand, as $f$ is also indecompsable over $\FF_q$, applying \cref{lem:equidistributed} to $f$ shows that $f(\mathbf{U}_{\FF_q^{n+1}})$ is $\eps'$-close to $\mathbf{U}_{\FF_q}$.
Therefore, the statistical distance between $f(G(\mathbf{U}_S))$ and $f(\mathbf{U}_{\FF_q^{n+1}})$ is at most
\begin{equation}\label{eq:distance}
(2^{d-1}\delta+\eps')+\eps'=2^{d-1}\delta+2\eps'=C_0 2^{d-1}(2d-1)/q+2C_1 d^2/\sqrt{q}    
\end{equation}
which is bounded by $\eps$ provided that $q\geq C(d2^d/\eps+d^4/\eps^2)$ and $C>0$ is a large enough absolute constant. 
The seed length of $G$ is 
\[
2\log|T|+2\log q=O(d\log n+\log(1/\delta)+\log q)=O(d\log n+\log q)\]
as $\delta=C_0(2d-1)/q$.
\end{proof}

We conclude this section by proving \cref{thm:prime-degree-intro}, which states that the required field size can be improved to $O(d^4/\eps^2)$ if we only want to fool polynomials of prime degree.

\begin{thm}[{\cref{thm:prime-degree-intro} restated}]\label{thm:prime-degree}
There exists an absolute constant $C>0$ such that for $\eps>0$ and $q\geq C(d^4/\eps^2)$ with $\mathrm{char}(\FF_q)\geq d(d-1)+1$,
$G$ as in \cref{construction:prg} is a PRG for $(n+1)$-variate polynomials of \emph{prime} degree
up to $d$ with error $\eps$ and seed length $O(d\log n+\log q)$.
\end{thm}

\begin{proof}[Proof Sketch]
Let $f\in\FF_q[\mathbf{x}, y]$ be a polynomial whose degree $d'$ is prime and at most $d$.
We want to prove that $f(G(\mathbf{U}_S))$ and $f(\mathbf{U}_{\FF_q^{n+1}})$ are $\epsilon$-close (in statistical distance).
Suppose $f$ is decomposable over $\FF_q$. Then $f=g(h)$ for some polynomials $g,h$ over $\FF_q$ where $\deg(g)\geq 2$, and as $d'=\deg(f)$ is prime, we must have $\deg(g)=d'$ and $\deg(h)=1$. In this case, the theorem follows by replacing $f$ with $h$, which has degree one, and applying \cref{thm:main2}.
So we may assume that $f$ is indecomposable over $\FF_q$.

The rest of the proof follows that of \cref{thm:main2}, except that we could bound the probability that $F$ is decomposable over $\FF_q$ by $2\delta$, rather than by $2^{d-1}\delta$, using the following observation:

In the application of \cref{lem:Bertinian}, the polynomial has the special form $g^*=f^*+c t$, where $f^*\in\FF_q[\mathbf{x},y]$, $c\in \FF_q^\times$, and $\deg(f^*)=d'$. By making the substitution $t\mapsto -c^{-1} t$, we may assume $c=-1$.
Consider any $\mathbf{a}\in\overline{\FF}_q^n$
such that $g^*_{\mathbf{a}}=g^*(a_1 x,\dots, a_n x, y)$ is reducible over $\overline{\FF_q(t)}^n$.
We claim that $g^*_{\mathbf{a}}$ factorizes into linear factors over $\overline{\FF_q(t)}$.
To see this, note that $f^*_{\mathbf{a}}=f^*(a_1 x,\dots, a_n x, y)$ is a decomposable polynomial over $\overline{\FF}_q$ of degree $d'$ by \cref{lem:indec2irred} and the fact that $g^*_{\mathbf{a}}=f^*_{\mathbf{a}}-t$ is reducible over $\overline{\FF_q(t)}$. So we may write $f^*_{\mathbf{a}}=\alpha(\beta)$ where $\alpha\in \overline{\FF}_q[z]$, $\beta\in \overline{\FF}_q[\mathbf{x},y]$, and $\deg(\alpha)>1$.
As $d'$ is prime, we must have $\deg(\alpha)=d'$ and $\deg(\beta)=1$. As $\alpha$ is univariate, $\alpha-t$ factorizes into linear factors $\alpha_1,\dots,\alpha_{d'}$ over $\overline{\FF_q(t)}$.
So $g^*_{\mathbf{a}}=\alpha(\beta)-t=(\alpha-t)(\beta)$ factorizes into the linear factors $\alpha_1(\beta),\dots,\alpha_{d'}(\beta)$
over $\overline{\FF_q(t)}$.

This observation shows that there is only one bad factorization pattern to rule out, namely, the complete factorization into linear factors.
This allows us to save a factor of $2^{d-1}-1$ and reduce the error probability in \cref{lem:reduction} from $(2^d-1)\delta+\delta$ to $\delta+\delta=2\delta$.
The bound on the statistical distance between $f(G(\mathbf{U}_S))$ and $f(\mathbf{U}_{\FF_q^{n+1}})$ in \eqref{eq:distance} now becomes $2\delta+2\epsilon'=2C_0 (2d-1)/q+2C_1 d^2/\sqrt{q}$, which is bounded by $\eps$ provided that $q\geq C(d^4/\eps^2)$ and $C>0$ is a large enough absolute constant.
\end{proof}

\section{Open Problems}
\label{sec:open}

We conclude with some open problems. The most obvious one is reducing the required field size in our construction. Using Bogdanov's \cite{Bogdanov05} paradigm, it seems necessary for the field to be of size at least polynomial in $d$, since this argument relies on the Weil bound (and indeed, as mentioned in  \cref{sec:intro}, the seed lengths of the known constructions over small fields like $\mathbb{F}_2$ are worse). Still, one could hope to obtain seed length $O(d \log n)$ with $q$ being polynomial in $d$, and not exponential in $d$. In our construction, $q$ is exponential in $d$ due to the need to apply a union bound over all possible vectors in $\{0,1\}^d$ characterizing the factorization pattern of $f_{\mathbf{a}}$. It could very well be that there is a more clever argument that rules out multiple vectors at once. We also mention again  \cref{question:N}. As explained in \cref{sec:lecerf}, improved upper bounds on the quantity $N^*(d,D,\mathbb{F})$ in that question would improve the field size required by our construction.

A related open problem is removing the requirement that the characteristic of $\mathbb{F}_q$ is at least $d(d-1)+1$. This requirement comes from using Lecerf's \cite{Lec07} arguments (dating back to Gao \cite{Gao2003} and Ruppert \cite{Ruppert86, Ruppert99}).

Finally, low-degree polynomials form a natural ``weak'' class of polynomials. However, rather than assuming bounds on the degree of polynomials, one can also consider other weak classes of polynomials, where the restriction comes from bounding their algebraic circuit complexity. This forms another interesting avenue for generalizing the results on PRGs for low-degree polynomials. As an analogy, in the context of Boolean computation, the problem of constructing explicit PRGs for weak computational classes (such as bounded-depth circuits or read-once oblivious branching programs) is well studied (see \cite{Vadhan12}). For algebraic computational models, however, much less is known. Most of the research in this area has focused on constructing \emph{hitting sets} of limited models of algebraic circuits (see \cite{SY10, Saxena09, Saxena14} for some surveys on this topic), due to the relation to the famous Polynomial Identity Testing Problem. To the best of our knowledge, there is no known explicit construction of PRGs for any natural class of algebraic computation. A concrete and intriguing open problem is to explicitly construct PRGs for the class of sparse polynomials, for which, as described in the references above, there are many known explicit constructions of hitting sets.

\section*{Acknowledgments} We thank Jesse Goodman and Pooya Hatami for helpful discussions. Part of this work was carried out while the first two authors were visiting the Simons Institute for the Theory of Computing at UC Berkeley. We thank the institute for their support and hospitality.

\bibliographystyle{alpha}
\bibliography{refs}

\end{document}